\documentclass[aps,pra,twocolumn,superscriptaddress,nofootinbib]{revtex4-2}

\usepackage{graphicx}
\usepackage{amsmath,amssymb,amsfonts}
\usepackage{bm}
\usepackage{amsthm}
\usepackage{quantikz}
\usepackage{hyperref}
\usepackage{mathrsfs}
\usepackage{braket}
\usepackage{xcolor}
\usepackage{bbold}
\usepackage{textcomp}
\usepackage{booktabs}
\usepackage{algorithm}
\usepackage[noend]{algpseudocode}
\usepackage{algorithmicx}
\usepackage{listings}

\newtheorem{theorem}{Theorem}

\newtheorem{lemma}[theorem]{Lemma}
\newtheorem{corollary}[theorem]{Corollary}
\newtheorem{definition}{Definition}

\begin{document}

\title{Efficient Many-Body Shadow Metrology via Clifford Lensing}

\author{Sooryansh Asthana}
\affiliation{Department of Physics, Indian Institute of Technology Bombay, Maharashtra 400076, India}

\author{Conan Alexander}
\affiliation{Department of Physics and NMR Research Center, Indian Institute of Science Education and Research, Pune 411008, India}

\author{Anubhav Kumar Srivastava}
\affiliation{ICFO - Institut de Ciencies Fotoniques, The Barcelona Institute of Science and Technology, Castelldefels (Barcelona) 08860, Spain}

\author{T. S. Mahesh}
\affiliation{Department of Physics and NMR Research Center, Indian Institute of Science Education and Research, Pune 411008, India}

\author{Sai Vinjanampathy}
\email{sai@phy.iitb.ac.in}
\affiliation{Department of Physics, Indian Institute of Technology Bombay, Maharashtra 400076, India}
\affiliation{Centre of Excellence in Quantum Information, Computing, Science and Technology, IIT Bombay, Maharashtra 400076, India}

\begin{abstract}
Quantum probes that enable enhanced exploration and characterization of complex systems are central to modern science, spanning applications from biology to astrophysics and chemical design. In large many-body quantum systems, interactions delocalize phase information across many degrees of freedom, dispersing it away from accessible measurements and limiting the scalability of quantum metrology.
Here we show that experimentally accessible Clifford operations acting jointly on quantum states and observables can refocus this distributed information. These operations implement what we term {\it Clifford lensing}--transformations that coherently localize phase information onto a reduced set of degrees of freedom, mapping optimal measurements onto observables of reduced Pauli weight.
We establish a correspondence between quantum error-correcting codes and interferometric constructions that enforce deterministic phase kickback, and generalize this to circuits that concentrate many-body phase information onto a controllable subset of qubits. We further develop partial shadow tomography protocols for estimating subsystem-supported phases.
We experimentally demonstrate these principles in liquid-state nuclear magnetic resonance systems of up to fifteen qubits, achieving optimal sensing with constrained resources.
Our results establish a scalable route to coherent control of information flow in interacting quantum systems, enabling many-body quantum sensing and multimode interferometry across complex architectures.
\end{abstract}

\maketitle

\section{Introduction}
\label{sec: Introduction}

Quantum advantage often requires many-body measurements. Examples include syndrome measurements in quantum error correction, collective measurements in communication protocols, and optimal measurements in quantum metrology \citep{chiaverini2004realization, jansen2022enumerating, Braunstein_Caves_94}. In quantum metrology, optimally chosen probe states and measurements can saturate the Cramér–Rao bound on the variance of an unbiased estimator, with ultimate sensitivity set by the quantum Fisher information (QFI) \citep{sidhu2020geometric}.
For general interacting many-body systems, however, the measurements required to achieve this bound—that is, projections onto the eigenstates of the symmetric logarithmic derivative (SLD)—are highly nonlocal and explicitly parameter-dependent \citep{Braunstein_Caves_94}. Such measurements involve nonlocal observables and are therefore infeasible on realistic platforms. Consequently, experimental quantum metrology has relied on intuition-driven measurement design, constrained by platform-specific limits and control resources \citep{ding2022enhanced, colombo2022time, yin2023experimental, deng2024quantum, guo2020distributed, giovannetti2011advances, PhysRevLett.116.053601}. Performance is assessed \textit{ex post facto}, by demonstrating advantage over classical strategies and proximity to theoretical bounds. This has created a persistent gap between formal optimality, captured by the SLD, and practical implementability.

In analogy with classical interferometry, where optical elements redistribute phase information across modes, one may ask whether a corresponding transformation exists in quantum systems that concentrates metrologically relevant information into accessible degrees of freedom. This motivates the following question: given a probe state and a phase-embedding generator, can one systematically transform the optimal but inaccessible SLD measurement into an implementable surrogate without sacrificing the QFI?
We answer this question in the affirmative by applying experimentally accessible Clifford operations jointly to the SLD and the phase-embedded quantum state—a procedure we term {\it Clifford lensing}. These transformations coherently localize phase information onto a reduced set of degrees of freedom while preserving expectation values and the QFI. Operationally, they redistribute phase information via phase kickback, thereby concentrating it onto a smaller number of qubits. As a result, the entanglement structure of the SLD is reorganized, yielding measurements that are experimentally accessible yet retain metrological optimality.

We quantify measurement complexity using Pauli weight. Under appropriate Clifford transformations, nonlocal correlations in the SLD and the phase-embedded state are reorganized so that phase information is concentrated onto fewer qubits, reducing the complexity of the measurement operator while preserving sensitivity. When the probe state and generator admit a stabilizer description, Clifford lensing can be performed analytically. In this regime, standard protocols such as Ramsey interferometry admit a natural interpretation in terms of stabilizer quantum-error-correcting codes. More generally, our framework yields families of metrology protocols inspired by error correction, in which new phase kickback schemes can be constructed by augmenting existing codes with Clifford operations.
Clifford lensing also has important implications for classical shadow tomography in metrology \citep{huang2020predicting}. Although classical shadows have been demonstrated across many platforms \citep{huggins2022unbiasing,gandhari2024precision,zhang2021experimental,stricker2022experimental,struchalin2021experimental,kim2025distributed}, the variance of shadow estimators scales exponentially with the maximum Pauli weight of the target operator \citep{huang2020predicting}, rendering direct estimation of eigenstates of the SLD infeasible. Clifford lensing addresses this obstacle by concentrating phase information into lower-weight observables compatible with efficient shadow estimation. In this sense, it concentrates non-Clifford resources onto a reduced subsystem, enabling scalable shadow-based metrology.

Two further challenges arise in implementing shadow tomography under realistic constraints: most sensing platforms offer only restricted unitary operations and measurement sets. We address this by identifying a class of quantum channels that preserve all information relevant to parameter estimation. We term these {\it metrologically sufficient} channels and show that, although not tomographically complete, they preserve the QFI. This enables efficient estimation of metrological observables without full state reconstruction. We experimentally demonstrate our framework on a liquid-state nuclear magnetic resonance (NMR) platform, performing shadow-based quantum metrology on systems of up to fifteen qubits. These results establish a concrete route from many-body SLD optimality to Clifford-accessible, shadow-compatible measurements, thereby narrowing the gap between theoretical optimality and experimental feasibility.

The remainder of this paper is organized as follows. In Sec.~\ref{Clifford lensing of SLDs}, we introduce Clifford lensing and analyze its effect on the SLD. In Sec.~\ref{Shadow with Ensemble Measurements}, we develop metrologically sufficient channels and partial shadow tomography protocols under restricted control. In Sec.~\ref{Sec: NMR Based shadow metrology}, we present an experimental realization using an NMR platform. We conclude in Sec.~\ref{Outlook}.

\section{Clifford lensing of  many-body symmetric logarithmic derivatives}
\label{Clifford lensing of  SLDs}
 {
Consider a local quantum metrology protocol defined by a pure $n$-party probe state
$\rho_\theta = U_\theta \rho_0 {\overline U}_\theta$ with parameter embedding unitary
$U_\theta = \exp({-i\theta G}),$ generated by $G$ and an initial state $\rho_0=\ket{\psi_0}\!\bra{\psi_0}$.
The ultimate precision for local estimation of $\theta$ is quantified by the {QFI}, $\mathcal F_Q[\rho_\theta]=\mathrm{Tr}(\rho_\theta \mathcal L_\theta^2),$
where $\mathcal L_\theta$ is the SLD.
Throughout, we define a \emph{metrology protocol} as the pair $\{\ket{\psi_0},G\}$ and assume that the optimal SLD measurement is implemented, so all protocols considered are Fisher-saturating. For pure states, $\mathcal L_\theta
=2\bigl(\ket{\partial_\theta\psi_\theta}\bra{\psi_\theta}
+\ket{\psi_\theta}\bra{\partial_\theta\psi_\theta}\bigr),$ with 
$\ket{\partial_\theta\psi_\theta}=-iG\ket{\psi_\theta}.$
At $\theta=0$, the SLD reduces to the {rank-two operator}
$\mathcal L
=2i\bigl(\ket{\psi_0}\bra{\chi}-\ket{\chi}\bra{\psi_0}\bigr),$
where the state $|\chi\rangle$ is given by
 $\ket{\chi}
={(G-\langle G\rangle_{\psi_0})\ket{\psi_0}}/
({{\mathrm{Var}_{\psi_0}(G)}})^{1/2}.$
The eigenstates of ${\cal L}$ are
$\ket{\lambda_\pm}=(\ket{\psi_0}\pm i\ket{\chi})/\sqrt{2}$,
which define the optimal measurement (see Appendix (\ref{App: Complexity of measurement of SLD})). Direct measurement in the $\{\ket{\lambda_\pm}\}$ basis is typically intractable for many-body probes due to the high Pauli weight of the corresponding projections.  

We now present our first main result, {\it viz.}, Clifford lensing of a metrology protocol. We apply a $\theta$-independent Clifford isometry $C$ prior to measurement, inducing the transformations $\mathcal L \;\mapsto\; \mathcal L_C = C \mathcal L \overline{C} ,~~
\rho_\theta \;\mapsto\; C \rho_\theta \overline{C}.$  Clifford lensing corresponds to the existence of such an isometry $C$ that localizes all parameter dependence onto a reduced subsystem, i.e., for which, in a neighborhood of $\theta=0$,
\begin{align}
C\rho_\theta {\overline C}=\rho^{(k)}_\theta \otimes \rho^{(n-k)}_{\rm aux}, ~~C{\cal L}{\overline C}=\mathcal L^{(k)} \otimes \mathcal {\cal L}^{(n-k)},
\end{align}
where  all dependence on the embedded phase  is confined to the $k$-party subsystem. The auxiliary state {$\rho^{(n-k)}_{\rm aux}$} is independent of $\theta$, satisfying $\partial_\theta \rho^{(n-k)}_{\rm aux} = 0.$ Consequently, the auxiliary degrees of freedom carry no metrological information, and the protocol is effectively reduced to a lower-dimensional metrology problem described by $\{\rho^{(k)}_\theta,\mathcal L^{(k)}\}$. While such a factorized structure does not arise universally, we show that it can be realized in broad and physically relevant classes of metrological protocols, where Clifford transformations reorganize the probe degrees of freedom so as to isolate the parameter dependence into a reduced effective subsystem. 

  Clifford conjugation preserves the Pauli group and thus a carefully chosen Clifford conjugation may enable a reduction of the Pauli weight of the SLD without affecting either {the QFI} or the optimality of the measurement.
Crucially, the SLD has support on a two-dimensional subspace,
{$\mathcal H_{\mathrm{SLD}}
=\mathrm{span}\{\ket{\psi_0},\ket{\chi}\},$}
so every protocol admits a reduction to a single {logical qubit} embedded in the full Hilbert space. Clifford lensing corresponds to choosing a Clifford isometry that localizes this logical qubit onto a small set of physical degrees of freedom.  Clifford lensing is distinct from the Clifford cooling protocols employed in  \citep{Martina_Frau_25, Fux2025, Masot-Llima24, Qian24, Huang25} which are solely designed to disentangle a quantum state via Clifford transformations. In contrast, Clifford lensing aims to localize the embedded phase to a smaller number of subsystems using Clifford gates. We show below that optimizing over Clifford gate sets separates metrology protocols into two families, {\it viz.},  stabilizer-compatible and stabilizer-incompatible. The classification is based on whether or not the eigenstates of the SLD and parametrized quantum state can be transformed to an effective single qubit system by Clifford lensing.  

}

\subsection{Stabilizer-compatible metrology protocols}
A metrology protocol $\{\ket{\psi_0},G\}$ is said to be \emph{stabilizer-compatible} if its phase information can be localized to a single qubit purely via Clifford transformations. In this case, the SLD subspace is mapped by Clifford conjugation to the {codespace of an $[n,1]$ stabilizer code}, with $G$ acting as a logical Pauli operator on the encoded qubit. The optimal SLD measurement can be implemented using Clifford circuits alone. This establishes a direct correspondence between stabilizer-compatible metrology protocols and deterministic Clifford phase kickback, which we formalize in the following theorem.
\begin{theorem}[Deterministic Clifford phase kickback]
Let $\mathcal{H}_P = (\mathbb{C}^2)^{\otimes m}$ be a physical Hilbert space.
The following statements are equivalent:
\begin{enumerate}
\item \textbf{Deterministic Clifford phase kickback.}
There exist orthonormal states $\ket{\phi_0}, \ket{\phi_1} \in \mathcal{H}_P$,
a physical Pauli operator $P$, and a Clifford unitary $C$ such that,
for all $\alpha, \beta \in \mathbb{C}$ and all $\theta \in \mathbb{R}$,
\begin{equation}
C e^{-i \theta P}
\left( \alpha \ket{\phi_0} + \beta \ket{\phi_1} \right)
=
\left( e^{-i \theta Z} (\alpha \ket{0} + \beta \ket{1}) \right)
\otimes \ket{\mathrm{aux}},
\label{eq:clifford-kickback}
\end{equation}
where $\ket{\mathrm{aux}}$ is a fixed auxiliary state independent of
$\alpha$, $\beta$, and $\theta$.

\item \textbf{Encoded Pauli logical qubit.}
The subspace $\mathcal{H}_L = \operatorname{span}\{\ket{\phi_0}, \ket{\phi_1}\}
\subset \mathcal{H}_P$
 defines a two-dimensional codespace supporting a logical Pauli algebra
$\{X_L, Z_L\}$, such that the physical Pauli operator $P$ acts as the logical
operator $Z_L$ on $\mathcal{H}_L$ (up to a stabilizer).
Moreover, $\mathcal{H}_L$ arises as the codespace of a stabilizer code encoding
a single logical qubit, possibly up to gauge degrees of freedom.
\end{enumerate}
\end{theorem}
The proof of this theorem has been presented in Appendix (\ref{App: PLCC of stabilizer-compatible metrology protocols}).  In a metrology protocol, the unknown parameter $\theta$ is embedded through the unitary evolution $\exp(-i\theta G)$ acting on an initial probe state $\ket{\psi_0}$. Theorem 1 implies that whenever deterministic Clifford phase kickback is possible, this parameter encoding can be reduced to an effective single-qubit process.

We illustrate stabilizer-compatible Clifford lensing using paradigmatic Ramsey metrology with an $n$-qubit GHZ probe evolving under the collective generator $S_z$. The dynamics is confined to the repetition-code subspace encoding a single logical qubit, on which the generator acts as $S_z \mapsto (n/2) Z_L$.
A Clifford isometry composed of CNOT gates deterministically disentangles this codespace, mapping the entire $\theta$-dependence onto one physical qubit.
This realizes Clifford phase kickback, coherently amplifying the phase by a factor $n$ and yielding Heisenberg-limited scaling ${\mathcal F}_Q\sim n^2$ (see Appendix (\ref{App:PLCC of Ramsey metrology}) for details). In this sense, standard interferometric protocols such as Ramsey metrology can be viewed as instances of Clifford lensing, where phase information is coherently refocused onto a single measurable degree of freedom.

Next, we present a distinct sensing protocol constructed from first principles within our framework,  based on the logical states of the surface code. The logical code states are defined as ~\citep{marques2022logical},
\begin{align}
    |0\rangle_L \equiv \frac{1}{\sqrt{2}}(|0000\rangle + |1111\rangle),~~  
    |1\rangle_L \equiv\frac{1}{\sqrt{2}}(|0101\rangle + |1010\rangle).\nonumber
    \end{align}
Consider an $n$-logical-qubit GHZ state,
    $|\mathrm{GHZ}\rangle_L = {(|0\rangle_L^{\otimes n} + |1\rangle_L^{\otimes n})}/{\sqrt{2}},$
which can be prepared using the Clifford circuit (see details in Appendix (\ref{App:PLCC of surface code based metrology protocol})). One choice of logical $Z$ operator of the code may be given as $Z_L \equiv Z_1Z_2$ or $Z_L \equiv Z_3Z_4$. Employing the identity $CNOT_{i, i+1}(Z_iZ_{i+1}){CNOT_{i, i+1}} =Z_{i+1}$, the logical $Z$ operators of the code may be transformed to  metrologically relevant single-qubit operators acting on a subset of physical qubits. Hence, a phase $\theta$ may be embedded by subjecting this subset of physical qubits of the code to the phase evolution generated by $Z_i$, yielding the state
    ${(|0\rangle_L^{\otimes n} + e^{i n \theta} |1\rangle_L^{\otimes n})}/{\sqrt{2}}.$ The eigenstates of SLD are $(|0\rangle_L^{\otimes n}\pm i|1\rangle_L^{\otimes n})/\sqrt{2}$. These states can be disentangled by applying the inverse Clifford circuit, which coherently kicks back the accumulated phase $n\theta$ onto the first logical qubit. This process can be described entirely within the {stabilizer tableau formalism}. In this picture, applying the inverse Clifford circuit updates the tableau so as to disentangle the state while coherently kicking back the accumulated phase $n\theta$ to the first qubit.  

We emphasize that theorem 1 goes further than merely establishing error correctable quantum metrology schemes \citep{kessler2014quantum,arrad2014increasing,zhou2018achieving}. Our goal was the \textit{a priori} reduction of the maximum Pauli weight of the optimal measurement, and we relate every $[n,1]$ error-correcting code to an interferometer design that kicks the phase back to one qubit. Hence, generalizing the spirit of Ramsey metrology, we derive the family of interferometers for which one-body measurements suffice to saturate Fisher scaling. Based on this intuition, we also develop a general theory to optimize those protocols that lie outside this restricted set, which we present below.

\subsection{Stabilizer-incompatible metrology protocols}
 For generic choices of the generator $G$ and initial state $\ket{\psi_0}$, the extremal eigenstates $\ket{\lambda_\pm}$ of the SLD are themselves {non-stabilizer states}. Consequently, the optimal measurement necessarily requires non-Clifford resources. In this stabilizer-incompatible regime, Clifford lensing  can at best provide phase localization of the SLD to fewer qubits, by partially localizing phase information onto fewer qubits, thereby mapping the SLD to lower Pauli weight operators. We formalize this in theorem (\ref{thm:plcc_characterization_1}).

\begin{theorem}[Characterization of Clifford lensing]
\label{thm:plcc_characterization_1}
An $n$-qubit metrology protocol $\{\ket{\psi_0},G\}$ admits Clifford lensing if and only if there exist
a $\theta$-independent Clifford isometry $C:\mathcal{H}_{2^n}\to\mathcal{H}_{2^k}\otimes\mathcal{H}_{2^{n-k}}, k<n,$
a $k$-qubit probe state $\ket{\psi^{(k)}_0}$, an $(n-k)$-qubit auxiliary state
$\ket{\phi^{(n-k)}_{\mathrm{aux}}}$, and Hermitian operators
$G^{(k)}$ and $G^{(n-k)}_{\mathrm{aux}}$ such that
 $C\ket{\psi_0}
=
\ket{\psi^{(k)}_0}\otimes\ket{\phi^{(n-k)}_{\mathrm{aux}}},
C G \overline{C}
=
G^{(k)} \otimes G^{(n-k)}_{\mathrm{aux}},$
with the auxiliary state $\ket{\phi^{(n-k)}_{\mathrm{aux}}}$ an eigenstate of
$G^{(n-k)}_{\mathrm{aux}}$.
\end{theorem}

The proof of this theorem is shown in Appendix (\ref{App: hm:plcc_characterization_1}).
As an illustration of Theorem~(\ref{thm:plcc_characterization_1}), we now present metrology protocols that admit phase kickback onto  $k>1$ qubits rather than  a single qubit. We begin with a protocol defined by the fiducial state $\ket{\psi_0}
=
\prod_{j=2}^n \mathrm{CNOT}_{1\to j}
\prod_{i=1}^k (T_i H_i)
\ket{0}^{\otimes n},$  where the first qubit is entangled with the remaining $n-1$ qubits by CNOT gates, and the only non-Clifford resources are applied to the first $k$ qubits. This construction distributes the non-Clifford metrological resource across a 
$k$-qubit subsystem while the remaining qubits are prepared using Clifford operations alone. After phase embedding under the collective generator $S_z$, the Clifford part of the input circuit can be undone. This reverses the initial entangling map, causing the last $n-k$ qubits to deterministically disentangle and factor out of the state, while all $\theta$-dependence is confined to the first $k$ qubits. In this way, the phase information is localized onto a $k$-qubit register and the QFI is ${\cal O}(n^2)$.

Another example of Theorem~\ref{thm:plcc_characterization_1} is obtained by specifying a code space defined by a set of mutually commuting projectors $\{P_\alpha\}_{\alpha=1}^{n-k}$ acting on an $n$-qubit Hilbert space \citep{Bravyi10, movassagh2024constructing}. The joint eigenspace of these projectors supports the non-Clifford resource and has effective dimension $2^k$, so that its magic can be concentrated onto $k$ qubits. A Clifford isometry $C$ implements this concentration, yielding the decomposition $C\ket{\psi_0}=\ket{\tilde\psi_0}\otimes\ket{\phi_{\mathrm{aux}}}$ and $CG\overline C=\sum_i\tilde G^{(i)}\otimes G^{(i)}_{\mathrm{aux}}$, with $\ket{\phi_{\mathrm{aux}}}$ an eigenstate of $G^{(i)}_{\mathrm{aux}}$. The phase is therefore encoded entirely in the $k$-qubit logical subsystem that carries the magic. The QFI is governed by $\mathrm{Var}(\tilde G)$, which may be chosen to yield ${\cal O}(n^2)$ scaling (see Appendix (\ref{App: PLCC of commuting projection based quantum metrology protocols}) for details).

The last obstacle to many-body quantum metrology is to make $k$-body measurements. We employ shadow tomography \citep{huang2020predicting}, which samples the density matrix and reconstructs operator expectation values. For randomized measurements to exhibit quantum sensitivity, the variance of the SLD must satisfy $\mathrm{var}({\cal L})\lesssim\mathcal{F}_Q={\cal O}(n^2)$. This requires Pauli words in the SLD to be only logarithmically local, so that  $k={\cal O}(\log n)$ for random Pauli measurements, a condition achievable via Clifford lensing of ${\cal L}$ over a Clifford gateset.  Additionally, shadow tomography requires gates and measurements that are inaccessible by most experimental platforms. From a metrological perspective, we show that such restrictive gates and measurements can be circumvented, retaining exactly the information needed to achieve Fisher-saturating precision.

\section{Metrologically sufficient channels}
\label{Shadow with Ensemble Measurements}
 A channel $\mathcal{E}$ is \emph{metrologically sufficient} if it preserves the subspace $\mathrm{span}\{\rho_\theta,\partial_\theta\rho_\theta\}$, thereby {preserving the QFI}. For unitary phase embedding, this holds if and only if  $\mathcal{E}([G,\rho_\theta])=[G_{\rm eff},\mathcal{E}(\rho_\theta)]$ for some effective generator $G_{\rm eff}$. In contrast to this, standard shadow tomography aims at reconstructing generic observables and therefore  relies on single-shot, qubit-resolved measurements combined with random local unitaries \citep{huang2020predicting}. Such requirements are incompatible with many experimental platforms lacking individual addressability or arbitrary local Clifford control. Even when Haar-random unitaries are replaced by Clifford ensembles—sufficient for reproducing low-order moments—the assumption of local measurements remains prohibitive. Motivated by the observation that quantum metrology requires only the preservation of the SLD subspace rather than full tomographic completeness, we introduce a framework of partial shadow channels with restricted gates and measurements. Although not tomographically complete, these channels are metrologically sufficient whenever the SLD lies within the linear span of observables accessible under collective control. This provides a scalable approach to quantum metrology compatible with ensemble measurements and realistic experimental constraints. Below, we present a shadow-tomographic protocol that involves collective control and ensemble measurements.

\subsection{Partial shadow tomography with collective measurements}
Consider an $n$-qubit probe state $\rho_\theta$ subject to collective unitary control and ensemble-averaged measurement of collective observables. In each experimental run, a unitary $U$ is sampled from an implementable ensemble $\mathcal{U}$ and applied as $\rho_\theta \mapsto U \rho_\theta \overline{U}$. A collective observable $M$ is then measured, yielding the ensemble-averaged outcome
$\langle M \rangle_U = \mathrm{Tr}(M\,U\rho_\theta \overline{U})$.

Each data pair $(U,\langle M \rangle_U)$ is mapped to a classical shadow via a linear reconstruction map $\mathcal{R}_U$, producing an estimator $\hat{\rho}_U=\mathcal{R}_U(\langle M \rangle_U)$. The reconstruction is chosen such that, for all observables $O$ in a designated operator subspace $\mathcal{S}$, the estimator is unbiased:
$\mathbb{E}_{U}[\mathrm{Tr}(O\,\hat{\rho}_U)] = \mathrm{Tr}(O\,\rho_\theta)$.
Given $N$ independent realizations $\{U_i,\langle M \rangle_{U_i}\}_{i=1}^N$, expectation values are estimated as
$\widehat{\langle O\rangle} = N^{-1}\sum_{i=1}^N \mathrm{Tr}(O\,\hat{\rho}_{U_i})$. The protocol is metrologically sufficient if the SLD {$\mathcal{L}_\theta \in \mathrm{span}(\mathcal{S})$}, ensuring that all observables required to saturate the quantum Cramér--Rao bound can be reconstructed from the shadow data.

\subsection{Partial shadow tomography with collective unitaries}
As an example of restricted operations, we consider a protocol based on random collective Clifford unitaries $U^{\otimes n}$ and ensemble measurements of a collective observable. While the resulting measurement channel is not tomographically complete, it can be metrologically complete, enabling faithful parameter estimation whenever the {SLD} lies within the linear span of observables accessible under collective control (see details in Appendix  (\ref{app:metrological_sufficiency})). 

Operationally, Clifford lensing consists of applying a {Clifford isometry ${C}$} prior to measurement, followed by standard random Pauli measurements and classical shadows post-processing. The improvement arises entirely from the reduced shadow norm of the Clifford lensed observable. Clifford lensing directly translates into a reduction in sample complexity for suitably chosen observables. We next demonstrate how this Clifford lensing advantage can be realized experimentally via collective Clifford-inverted shadow measurements in an NMR-based metrological protocol.

\section{NMR Based shadow metrology}
\label{Sec: NMR Based shadow metrology}
Our theoretical results construct a family of Clifford lenses that coherently focus phase information and leverage controlled randomness within a restricted experimental setting to enable its efficient extraction. To experimentally validate this framework, we present a demonstration of a metrologically sufficient quantum channel, realized through controlled phase kickback and collective Pauli twirling–based shadow tomography.
 The protocol is implemented on a {star-topology register} of hexamethyl phosphoric acid triamide (HMPA), consisting of a central $^{31}\mathrm{P}$ spin coupled to 18 equivalent $^{1}\mathrm{H}$ spins enabling effective system sizes ranging from $n=1$ to $15$ qubits (shown in Fig.  (\ref{fig:experimentalschematic})). Our implementation exploits coherent phase accumulation on the central spin to embed the phase $\theta$, enabling efficient extraction of unitary features of the underlying dynamics. We implement a randomized version of Ramsey interferometry, wherein the goal of the experiment is to estimate the embedded phase in the state $|\psi(\theta)\rangle = (|0\rangle^{\otimes n}+e^{i n \theta}|1\rangle^{\otimes n})/{\sqrt{2}}$. The phase $\theta$ is estimated by creating a classical shadow implemented only with collective interactions from the set $Cl(2)^{\otimes n}$, which is shown to be metrologically sufficient in the Appendix (\ref{App: Theory of experimental shadow metrology using NMR}). The Clifford gates are implemented via  {gradient ascent pulse engineering (GRAPE)}, while {Carr-Purcell-Meiboom-Gill (CPMG)} pulses suppress practical decoherence during the signal acquisition. Data is recorded every $\Delta \theta = 0.2^{\circ}$ steps from $9.4^{\circ}-10.6^{\circ}$ and the experimental variance of the shadow-measured phase is retrieved using a simple error-propagation formula with no additional fitting parameters. The resultant plots of the experimental variance in the local metrology setting in the range indicated is shown in a log-log plot against the system size ($n$), with additional guiding lines for the Heisenberg limit (HL) and standard quantum limit (SQL) in Fig. \ref{fig:Shadow_experiment_Scaling}. As we see, {almost all the estimated shadow means  follow the Heisenberg line}, and the variance is suppressed by the large sample size of $10^{15}$, which is controlled by the HMPA concentration. On some of the received signal, there is leakage from gradient echoes.

{At large $n$ values, several effects which hinder precise phase readouts become significant, chiefly the Gaussian scaling with $n$ of signal to noise ratio, and the amplification of phase errors, proportional to $n$, accumulated throughout the pulse sequence. In order to mitigate these a number of measures were taken: The sample was maintained at an ambient temperature of 220 K which increased the signal to noise ratio significantly, the Clifford gates were optimized  using GRAPE \citep{khaneja2005optimal} for robustness with respect to control pulse inhomogeneity and the CPMG \citep{carr1954effects,meiboom1958modified} sequence was applied during all evolution intervals to reduce dephasing. Additionally the gradient filter was optimized to maintain strength of the desired signal peak while  preventing contamination from undesired coherence orders.}

\begin{figure}[hbt!]
    \centering
    \includegraphics[clip=,trim=0.3cm 0cm 0.3cm 0.3cm,width=\linewidth]{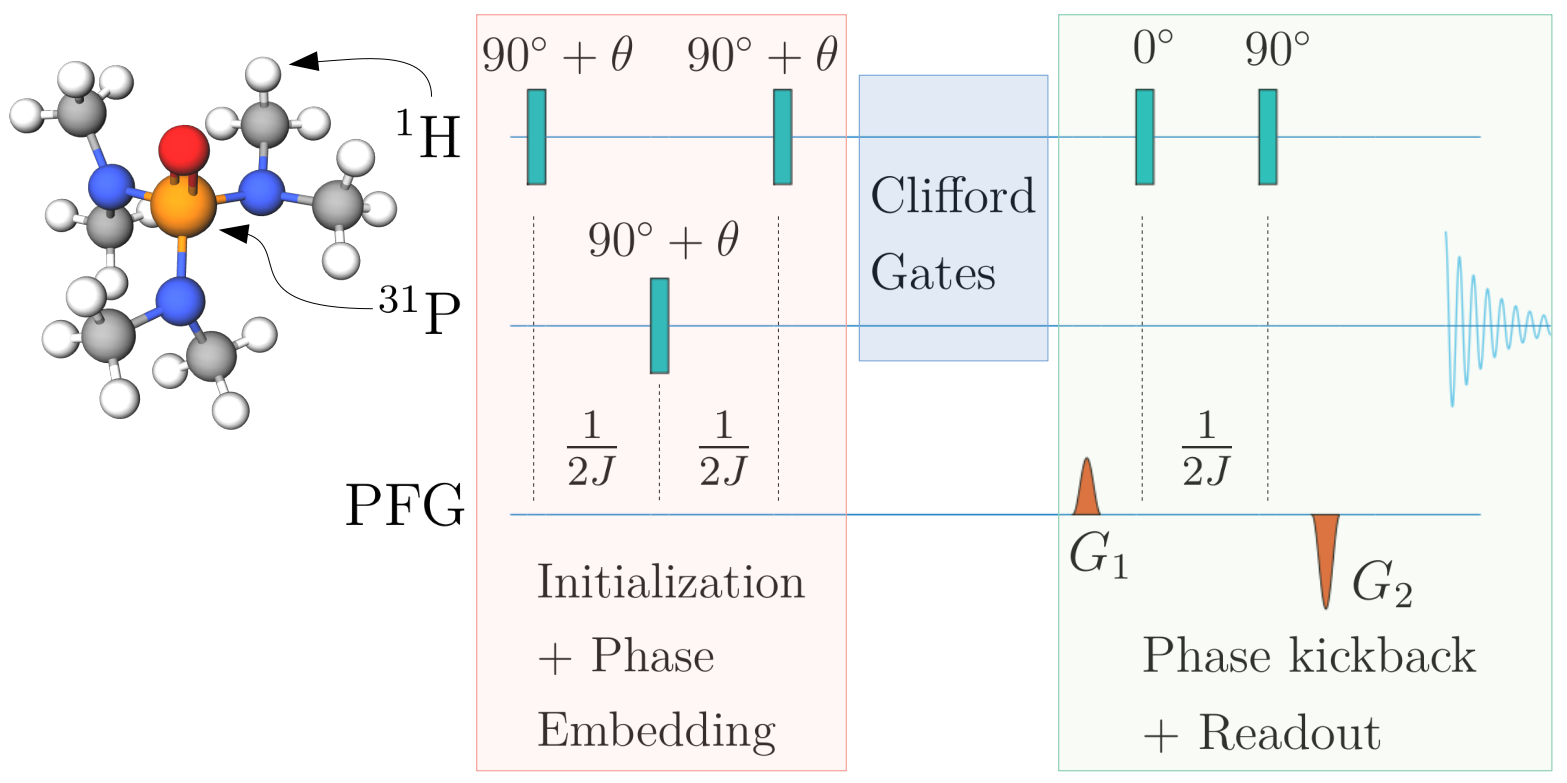}
    \caption{The HMPA molecular structure (left) and the NMR pulse sequence (right) for the metrology experiment. All pulses shown by rectangles are 90 degree rotations with phases as shown.  The Clifford gates were designed with GRAPE technique. The CPMG pulses during evolution delays are not shown.}
    \label{fig:experimentalschematic}
\end{figure}

\begin{figure}[hbt!]
        \centering
        \includegraphics[width=1\linewidth]{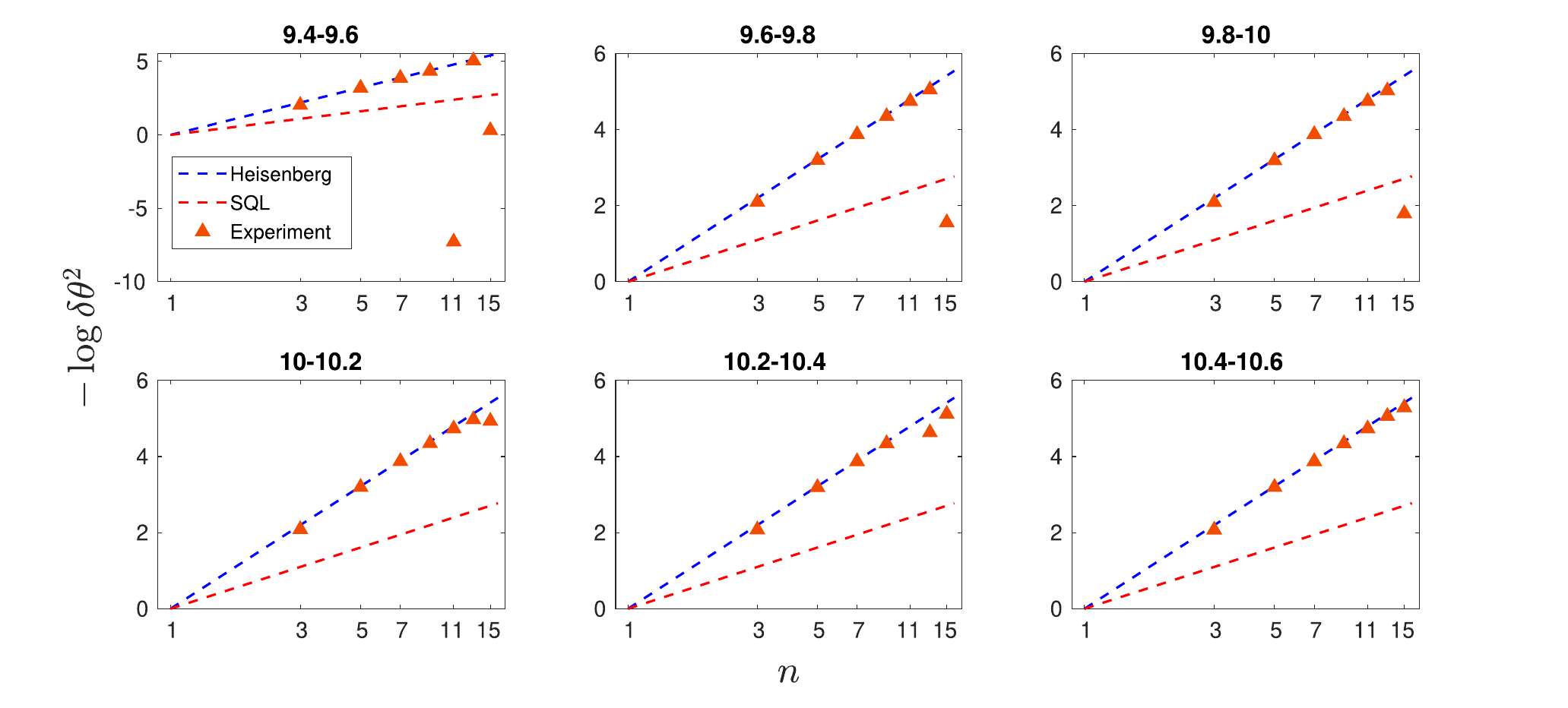}
        \caption{The experimental angular sensitivity vs coherence order $n$ at different ranges of embedded angles varying from $9.4^\circ$ to $10.6^\circ$. The blue dashed line represents the Heisenberg limit, while the red dashed line represents the shot-noise (standard quantum limit). The outliers at high coherence orders are due to a combination of gradient echoes, vulnerability of the weak NMR signal, and baseline errors in the NMR spectrum.
}
        \label{fig:Shadow_experiment_Scaling}
    \end{figure}

\section{Conclusion}
\label{Outlook}

In this work, we close the gap between theoretical optimality and experimental feasibility in quantum metrology by demonstrating that measurements achieving optimal sensitivity can be made experimentally viable through {\it Clifford lensing}. From an interferometric perspective, Clifford lensing provides a systematic way to design transformations that concentrate phase information into accessible measurement modes, enabling efficient implementation of many-body metrology protocols using classical-shadow techniques.
We introduced a class of metrologically sufficient protocols in which optimal sensitivity is retained despite restricted measurement control, and experimentally demonstrated this approach on a liquid-state NMR platform. Our results identify the Pauli weight of the measurement operator as the relevant notion of complexity in quantum metrology, and we propose this quantity as a quantitative measure of measurement complexity.
We further showed that quantum error-correcting codes can be interpreted as interferometers that enable phase kickback, establishing a general correspondence between encoding structures and metrological protocols. This perspective suggests that error-correcting codes can be systematically repurposed as sensing architectures, offering a new viewpoint on the interplay between quantum redundancy, many-body entanglement, and measurement.

An important direction for future work is the design of interferometric protocols based on alternative codes that localize phase information onto a reduced number of qubits. Another promising avenue is the development of task-specific shadow channels tailored to experimental constraints. Beyond metrology, such approaches may find applications in Hamiltonian learning, error mitigation, and sensing in constrained architectures, where the structure of the task—rather than full state reconstruction—determines the optimal measurement strategy.

\section*{Acknowledgments}

SV acknowledges motivating discussions at conferences ICTS/QuTr2025/01 and ICTS/pdcs2023/6. SV also acknowledges discussions with Marcello Dalmonte, David Elkouss, Gerald Fux, Pavithran Iyer and Paolo Zanardi.  CA and TS thank the National Mission on Interdisciplinary Cyber-Physical Systems for funding from the DST, Government of India, through the I-HUB Quantum Technology Foundation, IISER-Pune. AKS acknowledges support from the European Union’s Horizon 2020 Research and Innovation Programme under the Marie Sklodowska-Curie Grant Agreement No. 847517; MCIN/AEI (PGC2018–0910.13039/501100011033, CEX2019-000910-S/10.13039/501100011033, Plan National STAMEENA PID2022-139099NB, project funded by MCIN and by the “European Union NextGenerationEU/PRTR” (PRTR-C17.I1), FPI); Ministry for Digital Transformation and of Civil Service of the Spanish Government through the QUANTUM ENIA project call - Quantum Spain project, and by the European Union through the Recovery, Transformation and Resilience Plan – NextGenerationEU within the framework of the Digital Spain 2026 Agenda; CEX2024-001490-S [MICIU/AEI/10.13039/501100011033]; Fundació Cellex; Fundació Mir-Puig; Generalitat de Catalunya (European Social Fund FEDER and CERCA program); Barcelona Supercomputing Center MareNostrum (FI-2023-3-0024); European Union HORIZON-CL4-2022-QUANTUM-02-SGA – PASQuanS2.1, 101113690; EU Horizon 2020 FET-OPEN OPTOlogic, Grant No 899794; QU-ATTO, 101168628; EU Horizon Europe research and innovation program under grant agreement No. 101080086 NeQST.

\appendix

\section{Complexity of measurement of symmetric logarithmic derivative}
\label{App: Complexity of measurement of SLD}
In this appendix, we derive the SLD associated with a general unitary quantum metrology protocol in the local estimation regime and characterize the measurements that saturate the QFI. We also present the formulation in a form convenient for optimization over restricted measurement families,
such as those generated by low-complexity Clifford circuits.

We consider a metrological protocol defined by a pure probe state
$\ket{\psi_0}$ and a generator $G$. The unknown parameter
$\theta$ is embedded through unitary evolution
$U_\theta = \exp({-i\theta G}),
\ket{\psi_\theta} = U_\theta \ket{\psi_0}$,
giving the density operator
$\rho_\theta = \ket{\psi_\theta}\!\bra{\psi_\theta}.$ We work in the local estimation regime, where the parameter lies in a
small neighborhood of a reference value, taken without loss of generality
to be $\theta=0$. The task is therefore to distinguish infinitesimally
close states $\rho_\theta$ and $\rho_{\theta+d\theta}$, and all relevant
quantities are evaluated to leading order in $\theta$. The SLD operator $\mathcal L_\theta$ is implicitly defined through \citep{sidhu2020geometric} $\partial_\theta \rho_\theta
=
\left(
\rho_\theta \mathcal L_\theta
+
\mathcal L_\theta \rho_\theta
\right)/2.$ For unitary parameterizations of pure states,
$\partial_\theta \ket{\psi_\theta}
=
-i G \ket{\psi_\theta},$
and hence $\partial_\theta \rho_\theta
=
-i [G,\rho_\theta].$ The SLD therefore takes the explicit form
$\mathcal L_\theta
=
2\,\partial_\theta \rho_\theta
=
2i [\rho_\theta, G].$ Equivalently,
$\mathcal L_\theta
=
2(
\ket{\partial_\theta \psi_\theta}\bra{\psi_\theta}
+
\ket{\psi_\theta}\bra{\partial_\theta \psi_\theta}).$ Introduce a decomposition of the identity, $\mathbb{1}
=
\ket{\psi_\theta}\!\bra{\psi_\theta}
+
\sum_i \ket{\phi_i}\!\bra{\phi_i},$
where $\{\ket{\phi_i}\}$ spans the subspace orthogonal to
$\ket{\psi_\theta}$. Substituting this resolution gives
$\mathcal L_\theta
=
2i \sum_i
(
c_i^\ast(\theta)\ket{\psi_\theta}\!\bra{\phi_i}
-
c_i(\theta)\ket{\phi_i}\!\bra{\psi_\theta}),$ with coefficients
$c_i(\theta)
=
\bra{\phi_i} G \ket{\psi_\theta}.$ Defining the normalized orthogonal state
$\ket{\chi_\theta}
=
{
(G - \langle G \rangle_{\psi_\theta}) \ket{\psi_\theta}
}/{
\sqrt{\mathrm{Var}_{\psi_\theta}(G)}
},$
where $\langle G \rangle_{\psi_\theta}
=
\bra{\psi_\theta} G \ket{\psi_\theta},
\mathrm{Var}_{\psi_\theta}(G)
=
\bra{\psi_\theta} G^2 \ket{\psi_\theta}
-
\langle G \rangle_{\psi_\theta}^2,$
the SLD reduces to the rank-two form
$\mathcal L_\theta
=
2i \sqrt{\mathrm{Var}_{\psi_\theta}(G)}(
\ket{\psi_\theta}\!\bra{\chi_\theta}
-
\ket{\chi_\theta}\!\bra{\psi_\theta}).$
By construction,
$\braket{\chi_\theta|\psi_\theta}=0.$ Thus $\mathcal L_\theta$ acts nontrivially only on the
two-dimensional subspace
$\mathrm{span}\{\ket{\psi_\theta},\ket{\chi_\theta}\}.$ Within this subspace the SLD eigenstates are
$\ket{\lambda_\pm}
=
{
\ket{\psi_\theta} \pm i \ket{\chi_\theta}
}/{\sqrt{2}},$
with eigenvalues
$\lambda_\pm
=
\pm 2 \sqrt{\mathrm{Var}_{\psi_\theta}(G)}.$
All remaining eigenvalues vanish, so the SLD is traceless and of rank two.
For a POVM $\{\pi_i\}$, the QFI can be expressed as $\mathcal F_Q
=
\sum_i
{
\left(
\mathrm{Re}\,
\mathrm{Tr}[\rho_\theta \pi_i \mathcal L_\theta]
\right)^2
}/{
\mathrm{Tr}[\pi_i \rho_\theta]
}.$ For pure states this reduces to
$\mathcal F_Q
=
\sum_i
{
\left(
\mathrm{Re}\,
\bra{\psi_\theta} \pi_i \mathcal L_\theta
\ket{\psi_\theta}
\right)^2
}/{
\bra{\psi_\theta} \pi_i \ket{\psi_\theta}
}.$ The QFI is maximized by projective measurements onto the eigenstates of
$\mathcal L_\theta$. Defining the identity on the support of the SLD,
$\mathbb{1}_2
=
\ket{\psi_\theta}\!\bra{\psi_\theta}
+
\ket{\chi_\theta}\!\bra{\chi_\theta},$
the optimal measurement operators may be written as
$\pi_\pm
=
\left(
\mathbb{1}_2
\pm
{\mathcal L_\theta}/
({2\sqrt{\mathrm{Var}_{\psi_\theta}(G)}})
\right)/2.$

For unbiased estimators,
$\mathrm{Tr}(\rho_\theta \mathcal L_\theta)=0$, and the QFI simplifies to
$\mathcal F_Q
=
\bra{\psi_\theta} \mathcal L_\theta^2 \ket{\psi_\theta}
=
4\,\mathrm{Var}_{\psi_\theta}(G),$ 
recovering the standard result for unitary parameter estimation with pure
probe states.

\subsection{Need for Clifford lensing}
\label{Need for Clifford optimization}
Physically, the SLD ${\cal L}$ becomes highly non-local whenever the
QFI exhibits super-extensive scaling,
$\mathcal{F}_Q = {\cal O}(n^2)$. Such scaling signals that information about
the parameter is stored not in independent local degrees of freedom but
in collective many-body correlations distributed across the entire
system. Extracting this information, therefore, requires observables that
act coherently on a macroscopic fraction of the system. In Pauli-operator language, accessing these global correlations demands
operators containing high-weight Pauli strings acting on many qubits
simultaneously. Since the SLD corresponds to the optimal observable for
parameter estimation, it inherits this structure and typically involves
extensive Pauli-weight components. Consequently, measurements saturating
the QFI generally require collective many-body measurements rather than
local or few-body observables.

 Consider a scenario in which we employ  randomized gates to generate a classical shadow $\hat{\rho}$ of a quantum state $\rho$ for quantum phase estimation. {If we na\"{i}vely
 apply shadow tomography -- that is, use randomized measurements -- to estimate $\langle {\cal L} \rangle_{\hat{\rho}}$ and $\langle {\cal L}^2 \rangle_{\hat{\rho}}$, the statistical upper bound on the variance of ${\cal L}$ depends on the choice of random gates. For Clifford gates, the upper bound is given by $3~{\rm tr}({\cal L}^2)$.} 
 For an ensemble of $n$-qubit Pauli gates, the upper bound scales as $4^k ||{\cal L}||_{\infty}^2$, where $k$ denotes the locality of the SLD ${\cal L}$ \citep{huang2020predicting}.   On the other hand, quantum sensitivity is fundamentally upper bounded by QFI, satisfying ${\rm var}({\cal L}) \leq \mathcal{F}_Q$, which itself is upper bounded by $n^2$. For quantum sensitivity to be observed through randomized measurements, the statistical variance should be less than  variance of the SLD in the state. So, the following conditions must be met: (a) for the Clifford gates: $n^2 \gg 3~{\rm tr}({\cal L}^2)$.  (b) For $n$-qubit Pauli gates: $n^2 \gg 4^k ||{\cal L}||_{\infty}^2$.   From these conditions, we derive constraints on the locality $k$ of ${\cal L}$. For Clifford gates and Pauli gates, the locality $k$ must satisfy the conditions $k \ll \log\Big({n^2}/{3}\Big)$ and $k \ll \log n^2/3$ respectively.  For this reason, to achieve quantum sensitivity through randomized measurements, the locality of the SLD ${\cal L}$ should scale only logarithmically with $n$. This necessitates an optimization of ${\cal L}$ over some gate set, say $\{C\}$ so that its Pauli weight may be reduced and the reduced Pauli weight operator ${\cal L}' = C{\cal L}\overline{C}$ may be used for phase estimation in shadow-tomography-based quantum metrology. A further requirement of classical simulability of the gateset $\{V\}$ restricts it to be Clifford(2) gateset.  

Let $X$ denote a single-shot outcome of the shadow estimator, and let 
$\overline{X}_M$ be the empirical mean over $M$ independent measurement shots. 
Since $X \in [-B_{\rm Pauli}, B_{\rm Pauli}]$, Hoeffding’s inequality 
implies
\begin{equation}
\Pr\!\left[\,\big|\overline{X}_M - \mathbb{E}[X]\big| \geq \varepsilon\,\right] 
    \;\leq\; 2\exp\!\Bigg(-\frac{2M\varepsilon^2}{(2B_{\rm Pauli})^2}\Bigg).
\end{equation}
Here $B_{\rm Pauli}$ denotes the single-shot bound of the estimator, which for a Pauli operator of weight $k$ scales as
$B_{\rm Pauli} \sim 4^{k}\,\|\mathcal{L}\|_\infty$.
Therefore, to estimate $\mathbb{E}[X]$ within accuracy $\varepsilon$ and failure probability at most $\delta$, it suffices that~\citep{huang2020predicting}
\begin{equation}
\label{eq:hoeff_pauli}
M \;\geq\; 
\frac{2\cdot 4^{2k}\,\|\mathcal{L}\|_\infty^2}{\varepsilon^2}
\ln\!\frac{2}{\delta}.
\end{equation}

To resolve fluctuations at the intrinsic scale of the SLD, we set
$\varepsilon \;=\; \eta\,\sqrt{\mathrm{var}(\mathcal{L})},
 0 < \eta \ll 1,$
which yields the copy complexity
\begin{equation}\label{eq:N_pauli}
N_{\rm Pauli} \;\gtrsim\; 
\frac{2\cdot 4^{k}\,\|\mathcal{L}\|_\infty^2}{\eta^2\,\mathrm{var}(\mathcal{L})}
    \,\ln\!\frac{2}{\delta}.
\end{equation}

Let $\mathcal{L}' = C\,\mathcal{L}\,{\overline C}$ for a Clifford unitary $C$. 
Such conjugation preserves both the operator norm and the intrinsic variance,
$\|\mathcal{L}'\|_\infty = \|\mathcal{L}\|_\infty, 
 \mathrm{Var}(\mathcal{L}') = \mathrm{Var}(\mathcal{L}),$
while potentially reducing the effective Pauli weight from $k_{\rm pre}$ to 
$k_{\rm post}$. The corresponding Pauli shadow sample complexities scale as
\begin{align}
M_{\rm Pauli}^{\rm pre} &\;\gtrsim\; 
\frac{2\cdot 4^{2k_{\rm pre}}\,\|\mathcal{L}\|_\infty^2}
{\eta^2\,\mathrm{Var}(\mathcal{L})}
    \,\ln\!\tfrac{2}{\delta}, 
    \label{eq:N_pre}\\[4pt]
M_{\rm Pauli}^{\rm post} &\;\gtrsim\; 
\frac{2\cdot 4^{2k_{\rm post}}\,\|\mathcal{L}\|_\infty^2}
{\eta^2\,\mathrm{Var}(\mathcal{L})}
    \,\ln\!\tfrac{2}{\delta}.
    \label{eq:N_post}
\end{align}

In the worst-case scenario where 
$\mathrm{Var}(\mathcal{L}) \sim {\cal F}_Q \sim n^2$ 
(i.e., when the SLD saturates the quantum Fisher information),
Eqs.~\eqref{eq:N_pre}–\eqref{eq:N_post} reduce to
\begin{equation}
M_{\rm Pauli} \;\gtrsim\; 
\frac{2}{\eta^2\,n^2}\,
4^{2k}\,\|\mathcal{L}\|_\infty^2\,
\ln\!\tfrac{2}{\delta}.
\end{equation}
Thus, an unoptimised Pauli weight $k_{\rm pre}=\mathcal{O}(n)$ leads to an 
exponential copy complexity, whereas Clifford lensing reducing the 
locality to $k_{\rm post}=\mathcal{O}(\log n)$ ensures that the required number 
of measurement shots scales only polynomially with system size.

\section{Clifford lensing of SLDs}

\subsection{Clifford lensing of stabilizer-compatible metrology protocols}
\label{App: PLCC of stabilizer-compatible metrology protocols}
In this section, we present the proof of theorem (1) in the main text.
\setcounter{theorem}{0}
\begin{theorem}[Deterministic Clifford phase kickback]\label{thm:phase_kickback_equivalence_appendix}
Let $\mathcal{H}_P = (\mathbb{C}^2)^{\otimes m}$ be a physical Hilbert space. The following statements are equivalent:

\begin{enumerate}
\item \textbf{Deterministic Clifford phase kickback.}  
There exist orthonormal states $\ket{\phi_0}, \ket{\phi_1} \in \mathcal{H}_P$, a physical Pauli operator $P$, and a Clifford unitary $U_C$ such that, for all $\alpha,\beta \in \mathbb{C}$,
\begin{equation}
U_C \, e^{-i\theta P} \left( \alpha \ket{\phi_0} + \beta \ket{\phi_1} \right)
=
\left( e^{-i\theta Z} (\alpha \ket{0} + \beta \ket{1}) \right)
\otimes \ket{\mathrm{aux}},
\label{eq:deterministic_kickback}
\end{equation}
where $\ket{\mathrm{aux}}$ is a fixed auxiliary state independent of $\alpha$, $\beta$, and $\theta$.

\item \textbf{Encoded Pauli logical qubit.}  
The span $\mathcal{H}_L = \mathrm{span}\{\ket{\phi_0}, \ket{\phi_1}\} \subset \mathcal{H}_P$ defines a two-dimensional codespace supporting a logical Pauli algebra $\{X_L, Z_L\}$, with $P$ acting as $Z_L$ on $\mathcal{H}_L$ (up to a stabilizer), and with all other Pauli operators acting trivially on $\mathcal{H}_L$. Moreover, $\mathcal{H}_L$ arises as the codespace of a stabilizer code encoding one logical qubit, possibly up to gauge degrees of freedom.
\end{enumerate}
\end{theorem}

\begin{proof}
We prove the equivalence in both directions.

\noindent{\bf Proof of $(1) \Rightarrow (2)$: Phase kickback implies an encoded logical qubit}

Assume Eq.~\eqref{eq:deterministic_kickback} holds for all $\alpha,\beta$. By linearity, the states $\ket{\phi_0}$ and $\ket{\phi_1}$ span a fixed two-dimensional subspace $\mathcal{H}_L \subset \mathcal{H}_P$ that is preserved under the action of $\exp({-i\theta P})$ up to a global Clifford decoding. The action of $\exp({-i\theta P})$ on $\mathcal{H}_L$ induces a continuous one-parameter unitary group generated by a Hermitian operator acting nontrivially within $\mathcal{H}_L$. Since $P$ is a Pauli operator, this generator must act as a logical Pauli operator, which we identify as $Z_L$. Conjugation by the Clifford unitary $U_C$ maps the physical Pauli $P$ to a single-qubit Pauli $Z$ acting on the decoded logical qubit, while all remaining degrees of freedom factorize into the auxiliary state $\ket{\mathrm{aux}}$. Consequently, all Pauli operators commuting with $P$ act trivially on $\mathcal{H}_L$, forming an abelian stabilizer group whose joint $+1$ eigenspace is precisely $\mathcal{H}_L$. Thus, $\mathcal{H}_L$ is the codespace of a stabilizer code encoding one logical qubit, with $P$ realizing the logical operator $Z_L$. This establishes statement (2).

\noindent{\bf Proof of $(2) \Rightarrow (1)$: Encoded logical qubit implies phase kickback}

Conversely, let $\mathcal{C}$ be a stabilizer code encoding one logical qubit with logical Pauli operators $\{X_L, Z_L\}$ and codespace $\mathcal{H}_L = \mathrm{span}\{\ket{\phi_0}, \ket{\phi_1}\}$. By definition, there exists a Clifford decoding unitary $U_C$ such that $U_C \ket{b}_L = \ket{b} \otimes \ket{0}^{\otimes (m-1)}, b = 0,1.$

Let $P$ be a physical Pauli operator representing $Z_L$ on the codespace. Then, for all $\ket{\psi}_L \in \mathcal{H}_L$, $\exp({-i\theta P}) \ket{\psi}_L = \exp({-i\theta Z_L}) \ket{\psi}_L.$ Applying $U_C$ yields $U_C\exp({-i\theta P}) \ket{\psi}_L
=
\left( \exp({-i\theta Z}) \ket{\psi} \right)
\otimes \ket{0}^{\otimes (m-1)},$
which is precisely Eq.~\eqref{eq:deterministic_kickback}, with a fixed auxiliary state independent of $\alpha$, $\beta$, and $\theta$. This establishes deterministic phase kickback using Clifford operations only, proving statement (1).
\end{proof}
The theorem shows that deterministic, amplitude-independent phase localization under Clifford control is neither accidental nor generic: it is an operational signature of an encoded logical qubit with Pauli-valued logical operators. In metrological and distributed-sensing settings, this equivalence explains the natural emergence of GHZ- and repetition-type encodings as minimal stabilizer realizations permitting Clifford-only access to logical phases.
\subsubsection{Clifford lensing of Ramsey metrology}
\label{App:PLCC of Ramsey metrology}
We illustrate Clifford lensing of the paradigmatic Ramsey metrology protocol with an $n$-qubit GHZ probe. The probe is initialized in $\ket{\mathrm{GHZ}_n}=(\ket{0}^{\otimes n}+\ket{1}^{\otimes n})/\sqrt{2}$, which belongs to the two-dimensional repetition-code subspace $\mathcal H_L=\mathrm{span}\{\ket{0}^{\otimes n},\ket{1}^{\otimes n}\}$ encoding a single logical qubit. Parameter embedding is generated by the collective spin operator $S_z=\sum_{j=1}^n Z_j/2$, whose action is confined to $\mathcal H_L$. Within this subspace the generator acts as $S_z=nZ_L/2$, where $Z_L$ denotes the logical Pauli operator defined by $Z_L\ket{0}^{\otimes n}=\ket{0}^{\otimes n}$ and $Z_L\ket{1}^{\otimes n}=-\ket{1}^{\otimes n}$. Consequently, the encoded unitary reduces to a single-logical-qubit phase rotation with an enhanced generator. A $\theta$-independent Clifford isometry $C$, constructed as a sequence of CNOT gates from one qubit onto the remaining $n-1$, deterministically disentangles the repetition-code subspace, mapping $\ket{0}^{\otimes n}\mapsto\ket{0}\otimes\ket{0}^{\otimes(n-1)}$ and $\ket{1}^{\otimes n}\mapsto\ket{1}\otimes\ket{0}^{\otimes(n-1)}$. Under this transformation, the logical operator $Z_L$ is mapped to a single-qubit Pauli operator, $C Z_L {\overline C} = Z\otimes{(|0\rangle\langle 0|)}^{\otimes(n-1)}$, and the parameter-dependent unitary (in the subspace ${\cal H}_L$) becomes $C \exp({-i\theta S_z}) {\overline C} = \exp[{-in\theta Z/2}\otimes (|0\rangle\langle 0|)^{\otimes(n-1)}]$. Thus, all $\theta$-dependence is localized onto one physical qubit, while the remaining qubits are decoupled in a fixed ancillary state. This realizes Clifford phase kickback, whereby the accumulated phase is coherently amplified by a factor $n$, yielding Heisenberg-limited scaling of the quantum Fisher information, $\mathcal F_Q\sim n^2$.
\subsubsection{Clifford lensing of surface code based metrology protocol}
\label{App:PLCC of surface code based metrology protocol}
In this section, we provide details of the sensing protocol based on a distance-2 surface code.  We employ the construction of the four-qubit surface code, its logical basis states, and the associated logical Pauli operators introduced in \citep{marques2022logical} and adapt this encoding for quantum metrology within our framework.

The logical code space is defined on four physical qubits and spanned by the logical basis states
$\ket{0}_L \equiv \left(\ket{0000}+\ket{1111}\right)/\sqrt{2},
\ket{1}_L \equiv \left(\ket{0101}+\ket{1010}\right)/\sqrt{2}.$ Logical Pauli operators act within this codespace; in particular, one choice of logical $Z$ operator is
$Z_L \equiv Z_1 Z_2,$
with an equivalent representation given by $Z_L \equiv Z_3 Z_4$, reflecting the redundancy of logical operator realizations within the code. An $n$-logical-qubit Greenberger–Horne–Zeilinger (GHZ) state is defined as
$\ket{\mathrm{GHZ}}_L
= 
\left(
\ket{0}_L^{\otimes n} + \ket{1}_L^{\otimes n}
\right)/\sqrt{2}.$ This state can be prepared deterministically using a Clifford circuit acting on the physical qubits of the $n$ code blocks. Explicitly, one first prepares each block in $\ket{0}_L$ and then applies a sequence of logical Clifford gates that entangle the logical qubits in the standard GHZ pattern. Since all logical operations admit Clifford realizations on the physical qubits, the entire preparation circuit remains within the Clifford group and is efficiently describable using stabilizer methods.

For metrological purposes, it is advantageous to express logical generators in terms of single-qubit physical operators. This can be achieved using the Clifford identity $CNOT_{i,i+1}\,(Z_i Z_{i+1})\,CNOT_{i,i+1} = Z_{i+1},$
which allows a two-qubit $Z$ operator to be mapped onto a single-qubit $Z$ operator via conjugation by a CNOT gate.
Applying an appropriate sequence of CNOT gates within each code block transforms the logical operator $Z_L$ into a single-qubit Pauli operator acting on a designated physical qubit. As a result, the logical phase encoding can be implemented by applying a physical phase rotation generated by $Z_i$ on a subset of physical qubits—one per logical block.

Subjecting the selected physical qubits to the unitary evolution $\exp({-i \theta Z_i/2})$ induces an effective logical evolution generated by $Z_L$. Acting on the logical GHZ state, this yields $\ket{\psi_\theta}
=
\left(
\ket{0}_L^{\otimes n}
+
e^{i n \theta}
\ket{1}_L^{\otimes n}
\right)/\sqrt{2},$
where the phase $n\theta$ reflects coherent signal accumulation across the $n$ logical qubits. The SLD associated with this family of states has eigenstates
$\ket{\lambda_\pm}
=
\left(
\ket{0}_L^{\otimes n}
\pm i \ket{1}_L^{\otimes n}
\right)/\sqrt{2},$
which define the optimal measurement basis for estimating $\theta$.

To extract the accumulated phase, one applies the inverse of the Clifford circuit used to prepare the logical GHZ state. This inverse circuit disentangles the logical qubits while coherently ``kicking back'' the total phase $n\theta$ onto the first logical qubit. Consequently, the metrologically relevant information is mapped onto a single logical degree of freedom, enabling optimal readout with a local measurement. The entire protocol admits a compact description within the stabilizer tableau formalism. Preparation, phase encoding, and inverse Clifford evolution correspond to updates of the stabilizer and destabilizer generators. In particular, the inverse Clifford circuit updates the tableau so as to disentangle the stabilizers associated with multi-logical-qubit correlations, while preserving the accumulated logical phase as a relative phase on the first logical qubit. 
\subsection{Clifford lensing of stabilizer-incompatible metrology protocols}
\label{App: hm:plcc_characterization_1}
\begin{theorem}[Characterization of Clifford lensing]
\label{thm:plcc_characterization_1_proof}
An $n$-qubit metrology protocol $\{\ket{\psi_0},G\}$ admits Clifford lensing if and only if there exist
a $\theta$-independent Clifford isometry
$C:\mathcal{H}_{2^n}\to\mathcal{H}_{2^k}\otimes\mathcal{H}_{2^{n-k}}$, with $k<n$,
a $k$-qubit probe state $\ket{\psi^{(k)}_0}$, an $(n-k)$-qubit auxiliary state
$\ket{\phi^{(n-k)}_{\mathrm{aux}}}$, and Hermitian operators
$G^{(k)}$ and $G^{(n-k)}_{\mathrm{aux}}$ such that $C\ket{\psi_0}
=
\ket{\psi^{(k)}_0}\otimes\ket{\phi^{(n-k)}_{\mathrm{aux}}},
C G \overline{C}
=
G^{(k)} \otimes G^{(n-k)}_{\mathrm{aux}},$ 
with $\ket{\phi^{(n-k)}_{\mathrm{aux}}}$ an eigenstate of
$G^{(n-k)}_{\mathrm{aux}}$.
\end{theorem}

\begin{proof}
(\emph{Necessity})  
If the protocol admits Clifford lensing, there exists a $\theta$-independent Clifford isometry that preserves the measurement statistics associated with the SLD eigenbasis while rendering the optimal measurement local up to classical post-processing. This implies the existence of a Clifford isometry
$C:\mathcal{H}_{2^n}\to\mathcal{H}_{2^k}\otimes\mathcal{H}_{2^{n-k}}$
under which the SLD eigenbasis factorizes into an informative subsystem and a passive one. Consequently,
\(
C\ket{\psi_0}
=
\ket{\psi^{(k)}_0}\otimes\ket{\phi^{(n-k)}_{\mathrm{aux}}}.
\)
Preservation of the encoding statistics further requires the generator to decompose as
\(
C G \overline{C}
=
G^{(k)}\otimes G^{(n-k)}_{\mathrm{aux}}.
\)
Since the auxiliary subsystem carries no information about $\theta$, the state
$\ket{\phi^{(n-k)}_{\mathrm{aux}}}$ must be an eigenstate of
$G^{(n-k)}_{\mathrm{aux}}$.

\medskip
(\emph{Sufficiency})  
Conversely, suppose such a Clifford isometry, state decomposition, and generator factorization exist. Then the parameter embedding acts nontrivially only on the $k$-qubit probe, while the auxiliary subsystem remains in a fixed eigenstate and is therefore metrologically inert. Optimal measurements can thus be chosen on the probe subsystem alone, supplemented by arbitrary local measurements on the auxiliary qubits. As Clifford isometries map local Pauli measurements to local Pauli measurements up to classical post-processing, the full protocol admits Clifford lensing.
\end{proof}
\subsubsection{Clifford lensing of commuting projection based quantum metrology protocols}
\label{App: PLCC of commuting projection based quantum metrology protocols}

We elaborate on the commuting-projector construction underlying the example discussed in the main text. Consider an $n$-qubit Hilbert space $\mathcal H=(\mathbb C^2)^{\otimes n}$ and a family of mutually commuting projectors $\{P_\alpha\}_{\alpha=1}^{n-k}$ satisfying $P_\alpha^2=P_\alpha$ and $[P_\alpha,P_\beta]=0$. The code space is defined as the joint $+1$ eigenspace $\mathcal H_{\mathrm{code}}=\bigcap_{\alpha=1}^{n-k}\mathrm{Ker}(\mathbb 1-P_\alpha)$ and has dimension $\dim(\mathcal H_{\mathrm{code}})=2^k$. We assume that the probe state $\ket{\psi_0}$ lies entirely in this subspace, $P_\alpha\ket{\psi_0}=\ket{\psi_0}$ for all $\alpha$, and that the generator $G$ of the parameter encoding commutes with all projectors, $[G,P_\alpha]=0$, ensuring that the encoded state $\ket{\psi_\theta}=\exp({-i\theta G})\ket{\psi_0}$ remains confined to $\mathcal H_{\mathrm{code}}$.

Because the projectors commute, there exists a Clifford isometry $C:\mathcal H\to(\mathbb C^2)^{\otimes k}\otimes(\mathbb C^2)^{\otimes(n-k)}$ that simultaneously diagonalizes them, such that $CP_\alpha {\overline C}=\mathbb 1_{2^k}\otimes\ket{0}\!\bra{0}_\alpha$. Under this isometry, the code space is mapped to $(\mathbb C^2)^{\otimes k}\otimes\ket{\phi_{\mathrm{aux}}}$ with a fixed auxiliary stabilizer state $\ket{\phi_{\mathrm{aux}}}=\ket{0}^{\otimes(n-k)}$, and the initial probe decomposes as $C\ket{\psi_0}=\ket{\tilde\psi_0}\otimes\ket{\phi_{\mathrm{aux}}}$. All non-Clifford resource (magic) of the protocol is therefore localized on the $k$-qubit logical subsystem.

Since $G$ preserves the code space, it admits a block decomposition under $C$ of the form $CG{\overline C}=\sum_i \tilde G^{(i)}\otimes G^{(i)}_{\mathrm{aux}}$. Because $\ket{\phi_{\mathrm{aux}}}$ is an eigenstate of each $G^{(i)}_{\mathrm{aux}}$, with $G^{(i)}_{\mathrm{aux}}\ket{\phi_{\mathrm{aux}}}=\lambda_i\ket{\phi_{\mathrm{aux}}}$, the effective generator governing the parameter dependence is $\tilde G=\sum_i\lambda_i\tilde G^{(i)}$, and the encoded state factorizes as $Ce^{-i\theta G}{\overline C}(\ket{\tilde\psi_0}\otimes\ket{\phi_{\mathrm{aux}}})=e^{-i\theta\tilde G}\ket{\tilde\psi_0}\otimes\ket{\phi_{\mathrm{aux}}}$. Consequently, the quantum Fisher information is ${\cal F}_Q=4\,\mathrm{Var}(\tilde G)_{\tilde\psi_0}$.

\section{Metrologically sufficient channels}
\label{app:metrological_sufficiency}

Classical shadow tomography provides an efficient framework for predicting expectation values of a large family of observables using randomized measurements and classical post-processing. The shadow protocol introduced in \citep{huang2020predicting} is designed to be tomographically complete, enabling reconstruction of arbitrary observables at the cost of probing the full operator algebra. In quantum metrology, however, such completeness is unnecessary. The goal is instead to estimate an unknown parameter $\theta$ encoded in a family of quantum states $\rho_\theta =\exp({-i\theta G})\,\rho_0\,\exp({i\theta G}),$ 
where $G$ is a generator. Estimation performance is quantified by the QFI, which depends only on the local geometry of the statistical model in an infinitesimal neighborhood of $\theta$.
This motivates a weaker, task-oriented notion of completeness, tailored specifically to parameter estimation rather than full state reconstruction.

\subsection{Local Metrological Sufficiency}

\begin{definition}[Local metrological sufficiency]
A quantum preprocessing scheme—consisting of a physical quantum channel followed by measurement and classical post-processing—is said to be \emph{locally metrologically sufficient} at $\theta=0$ for the family $\{\rho_\theta\}$ if it preserves the local QFI,
$\mathcal{F}_Q(\rho_\theta)\big|_{\theta=0}
=
\mathcal{F}_Q(\mathcal{E}(\rho_\theta))\big|_{\theta=0}.$
\end{definition}

This notion is intentionally weaker than statistical sufficiency in the sense of Petz, as it requires preservation of a single local metric quantity rather than equivalence of full statistical experiments. For unitary parameterizations, the tangent operator at $\theta=0$ is $\dot{\rho}_0 := \partial_\theta \rho_\theta\big|_{\theta=0}
= -i[G,\rho_0].$
The local QFI depends only on the pair $(\rho_0,\dot{\rho}_0)$. Metrological sufficiency therefore requires faithful transmission of this tangent direction under the preprocessing scheme.


\subsection{Necessary Conditions for Metrological Sufficiency}
Let $\rho_\theta$ be a smooth one-parameter family of quantum states on a Hilbert space $\mathcal{H}$, with $\rho_0 := \rho_{\theta}\big|_{\theta=0}$. 
Let $G$ be a Hermitian operator (the generator of the parameter encoding), such that $\rho_\theta = e^{-i \theta G} \, \rho_0 \, e^{i \theta G}.$
Define the tangent operator at $\theta=0$ as
 $\dot{\rho}_0 := -i [G, \rho_0]$.
Let $\mathcal{M}$ be a quantum measurement channel, i.e., a completely positive trace-preserving (CPTP) map from operators on $\mathcal{H}$ to classical measurement outcomes (or, equivalently, to a classical-quantum state).
\begin{theorem}[Tangent-space injectivity]
A randomized measurement scheme associated with $\mathcal{M}$ can be locally metrologically sufficient at $\theta=0$ only if $\mathcal{M}(\dot{\rho}_0) \neq 0,$
i.e., $\mathcal{M}$ acts injectively on the tangent operator $\dot{\rho}_0$.
\end{theorem}

\begin{proof}
If $\mathcal{M}(\dot{\rho}_0)=0$, then
\[
\left.\frac{d}{d\theta} \mathcal{M}(\rho_\theta)\right|_{\theta=0} = 0,
\]
so the output state $\mathcal{M}(\rho_\theta)$ is locally insensitive to the parameter $\theta$ at first order.

As a consequence, the classical Fisher information (CFI) associated with the measurement outcomes vanishes at $\theta=0$. Since the quantum Fisher information (QFI) upper bounds the CFI and is attained only when the measurement is optimal, this implies that no nonzero QFI can be preserved under $\mathcal{M}$ locally. Hence, local metrological sufficiency is impossible.
\end{proof}

This injectivity condition is necessary but not sufficient, since metrological sufficiency additionally requires preservation of the full Fisher information metric, not just nontrivial parameter sensitivity.

\begin{corollary}[Two-twirl obstruction]
Let $\mathcal{U}$ be an ensemble of unitaries acting on $\mathcal{H}$, and let $\Phi^{(2)}_{\mathcal{U}}$ denote the associated second-moment (two-copy) twirling channel defined by
\[
\Phi^{(2)}_{\mathcal{U}}(X) := \int_{\mathcal{U}} (U \otimes U)\, X \,({\overline U} \otimes {\overline U})\, d\mu(U),
\]
where $d\mu(U)$ is the probability measure over the ensemble.

A necessary condition for local metrological sufficiency is
\[
\Phi^{(2)}_{\mathcal{U}}\big( \rho_0 \otimes \dot{\rho}_0 \big) \neq 0.
\]
\end{corollary}

This condition expresses that the second-moment structure of the unitary ensemble must preserve correlations between the state $\rho_0$ and its parameter-dependent tangent direction $\dot{\rho}_0$. If these correlations are erased by the twirling channel, then the measurement scheme cannot retain local metrological information.

\subsection{Measurements and Fisher Information}

Following preprocessing, a POVM $\{M_x\}$ produces outcome probabilities $p(x|\theta)$ with classical Fisher information $F_C(\theta)$.

\begin{theorem}[Measurement optimality]
For a fixed post-channel state $\sigma_\theta=\mathcal{E}(\rho_\theta)$, a measurement achieves ${\cal F}_C(\theta) = \mathcal{F}_Q(\sigma_\theta)$ 
if the POVM resolves the eigenbasis of the SLD of $\sigma_\theta$. Commutation with the SLD is sufficient but not necessary.
\end{theorem}

This optimality condition depends only on the final quantum state and is independent of the specific shadow reconstruction protocol.
\begin{proof}
For a POVM $\{M_x\}$, the outcome probabilities are
$p(x|\theta) = \mathrm{Tr}\big[ \sigma_\theta M_x \big].$
The classical Fisher information (CFI) is
$\mathcal{F}_C(\theta) = \sum_x {\big( \partial_\theta p(x|\theta) \big)^2}/{p(x|\theta)}.$
Using the definition of the SLD, we write
$\partial_\theta p(x|\theta) 
= \mathrm{Tr}\big[ (\partial_\theta \sigma_\theta) M_x \big]
=  \mathrm{Tr}\big[ (L_\theta \sigma_\theta + \sigma_\theta L_\theta) M_x \big]/2
= \mathrm{Re}\,\mathrm{Tr}\big[ \sigma_\theta L_\theta M_x \big].$ Substituting into the CFI,
$\mathcal{F}_C(\theta) 
= \sum_x {\big( \mathrm{Re}\,\mathrm{Tr}[\sigma_\theta L_\theta M_x] \big)^2}/{\mathrm{Tr}[\sigma_\theta M_x]}.$

By applying the Cauchy--Schwarz inequality in the Hilbert--Schmidt inner product weighted by $\sigma_\theta$, one obtains the quantum Cramér--Rao bound
$\mathcal{F}_C(\theta) \le \mathcal{F}_Q(\sigma_\theta),$
where
$\mathcal{F}_Q(\sigma_\theta) = \mathrm{Tr}\big[ \sigma_\theta L_\theta^2 \big].$

Equality holds if and only if the Cauchy--Schwarz inequality is saturated for all $x$, which occurs precisely when the measurement operators $M_x$ project onto subspaces that resolve the eigenbasis of $L_\theta$. In particular, this is achieved if $M_x$ are projectors onto the eigenvectors of $L_\theta$.

A sufficient (but not necessary) condition for this is that all $M_x$ commute with $L_\theta$, i.e., $[M_x, L_\theta]=0$, which ensures simultaneous diagonalizability. However, optimal measurements need only resolve the SLD eigenspaces and need not commute with $L_\theta$ globally.

Therefore, $\mathcal{F}_C(\theta) = \mathcal{F}_Q(\sigma_\theta)$ whenever the POVM resolves the eigenbasis of the SLD.
\end{proof}

This optimality condition depends only on the final quantum state $\sigma_\theta$ and its SLD $L_\theta$, and is independent of the specific protocol (e.g., shadow reconstruction) used to obtain measurement data.

Metrological sufficiency provides a principled relaxation of tomographic completeness, rooted in the local geometry of quantum statistical models. In the context of randomized measurements and classical shadows, this relaxation manifests as constraints on the second-moment structure of the unitary ensemble rather than invertibility on the full operator algebra.

These results indicate that \emph{metrologically sufficient shadow protocols} can achieve optimal parameter sensitivity with reduced sampling overhead, provided the generator-induced tangent space is preserved. This motivates the design of \emph{metrological shadows} tailored explicitly to sensing tasks rather than full state reconstruction.

\begin{theorem}[Shadow channel for ensemble measurements]
Let $M \in L(\mathbb{C}^d)$ be a traceless observable, $\operatorname{tr} M = 0$, and let $\mu_H$ denote the Haar measure on $U(d)$. Define the ensemble shadow channel $\mathcal{M}(\rho)
\equiv
\mathbb{E}_{U \sim \mu_H}
\big[\operatorname{tr}(U \rho \overline{U}\, O)\, \overline{U} O U \big].$
Then $\mathcal{M}$ admits the affine decomposition $\mathcal{M}(\rho) = c_{\mathbb{1}}\,\mathbb{1} + c_{\mathbb{F}}\,\rho,$
with coefficients
$c_{\mathbb{1}} = - \operatorname{tr}(O^2)/{(d(d^2-1))},
 c_{\mathbb{F}} =  \operatorname{tr}(O^2)/{d^2-1}.$
\end{theorem}
\begin{proof}[Proof]
By Haar invariance and Schur--Weyl duality, $\mathcal{M}$ acts diagonally on the identity and traceless subspaces, yielding the affine form 
$\mathcal{M}(\rho)=c_{\mathbb{1}}\mathbb{1}+c_{\mathbb{F}}\rho$.
The coefficients follow from evaluating $\mathcal{M}(\mathbb{1})$ and $\mathcal{M}({\mathbb F})$ using standard Haar integrals.
\end{proof}

\begin{corollary}[Invertibility and tomographic completeness]
The ensemble shadow channel $\mathcal{M}$ is invertible on $L(\mathbb{C}^d)$, with inverse $
\mathcal{M}^{-1}(X) = {(X - c_{\mathbb{1}}\,\mathbb{1})}/{c_{\mathbb{F}}}.$
Consequently, ensemble measurements of $\langle O \rangle_U$ are tomographically complete.
\end{corollary}

\begin{proof}[Proof]
Since $c_{\mathbb{F}}\neq 0$, $\mathcal{M}$ is invertible on $L(\mathbb{C}^d)$, with inverse obtained by solving the affine relation.
Invertibility implies tomographic completeness of ensemble measurements.
\end{proof}
\begin{definition}[Classical shadow for ensemble measurements]
Given a unitary $U$ drawn from $\mu_H$, the associated classical shadow estimator is defined as
 $\hat{\rho} \equiv \mathcal{M}^{-1}(\overline{U} O U).$
\end{definition}

\begin{lemma}[Unbiasedness of the shadow estimator]
The estimator $\hat{\rho}$ is unbiased, i.e., $\mathbb{E}_{U \sim \mu_H}[\hat{\rho}] = \rho.$
\end{lemma}

\begin{proof}[Proof]
By definition $\hat{\rho}=\mathcal{M}^{-1}(\overline{U}OU)$, and taking the Haar average yields
$\mathbb{E}[\hat{\rho}]=\mathcal{M}^{-1}(\mathcal{M}(\rho))=\rho$.
\end{proof}

\begin{corollary}[Unbiased estimation of observables]
For any observable $A \in L(\mathbb{C}^d)$, define
$\hat{a} \equiv \operatorname{tr}(A \hat{\rho}).$
Then $\hat{a}$ is an unbiased estimator of $\operatorname{tr}(A \rho)$, satisfying $\mathbb{E}[\hat{a}] = \operatorname{tr}(A \rho).$
\end{corollary}

\begin{proof}[Proof]
Linearity of the trace gives $\mathbb{E}[\hat{a}]=\operatorname{tr}(A\,\mathbb{E}[\hat{\rho}])=\operatorname{tr}(A\rho)$.
\end{proof}
\subsection{Shadow tomography with collective unitaries} In many experimental many-body platforms, unitary control is intrinsically collective: a given operation can be applied either identically to all qubits or not at all. When arbitrary elements of the full Clifford group $\mathrm{Cl}(2^n)$ are accessible, one recovers standard classical shadow tomography. Here we focus instead on the experimentally relevant regime in which only collective operations are available.

Let $\mathcal{S}\subset \mathrm{Cl}(2^n)$ denote the set of implementable unitaries. In platforms such as nuclear magnetic resonance, atomic ensembles, and cavity QED systems, physically realizable Clifford operations are restricted to collective single-qubit rotations, motivating the permutationally invariant subset $\mathcal{S}_{\mathrm{p}} := \{ U^{\otimes n} : U \in \mathrm{Cl}(2) \}.$
We consider a shadow-tomographic protocol based on random collective unitaries sampled from $\mathcal{S}_{\mathrm{p}}$, followed by ensemble measurements.

Formally, this protocol induces a randomized measurement channel
\begin{equation}
\mathcal{M}(\rho)
:=
\mathbb{E}_{U \sim \mu}
\!\left[
\operatorname{Tr}\!\left(M\, U^{\otimes n} \rho \, {\overline U}^{ \otimes n}\right)
\, {\overline U}^{ \otimes n} M\, U^{\otimes n}
\right],
\end{equation}
where $M$ is an experimentally accessible collective observable (e.g., total magnetization) and $\mu$ is a probability measure over the allowed unitaries. For continuous control, $\mu$ may be taken as the Haar measure on $SU(2)$; equivalently, the same channel is realized by sampling from a unitary $2$-design supported on $\mathrm{Cl}(2)$. Such collective shadows are not tomographically complete, owing to the restricted unitary control and ensemble-averaged readout. Nevertheless, they may be metrologically sufficient.

\subsection{Collective Clifford Classical Shadows with Collective $S_z$ Readout}
\label{app:collective_clifford_shadows}

In this appendix, we analyze a classical shadow tomography protocol based on collective Clifford twirling and a collective spin measurement. We construct the measurement channel in its correct form, exploit permutational and $\mathrm{SU}(2)$ symmetry, characterize its action on the collective operator algebra, and derive polynomial bounds on the variance of the reconstructed expectation values.

Let $\rho \in \mathcal L((\mathbb C^2)^{\otimes n})$ be an arbitrary $n$-qubit state.  
Let $\mu$ denote the uniform distribution over the single-qubit Clifford group $\mathrm{Cl}(1)$, and define the collective action
$U^{(n)} := U^{\otimes n},  \overline U^{(n)} := U^{\dagger \otimes n}.$
We consider measurements of the collective spin operator
$S_z := \sum_{i=1}^n Z_i/2.$
The measurement (shadow) channel is defined as
$\mathcal M(\rho)
:= \mathbb E_{U\sim\mu}
\Big[
\operatorname{Tr}\!\big(
U^{(n)} \rho\, \overline U^{(n)} S_z
\big)\,
\overline U^{(n)} S_z U^{(n)}
\Big].$
This map is linear, completely positive, and trace-non-increasing. It  corresponds to the linear classical-shadow channel associated with a collective observable readout, rather than a discrete-outcome POVM.
As in the general classical shadows formalism, the measurement channel admits a Choi--Jamiołkowski representation
$\mathcal M(\rho)
= \operatorname{Tr}_1
\!\left[
(\rho \otimes \mathbb I)\,
\Omega^{(2)}
\right],$
with the second-moment operator $\Omega^{(2)}
:= \mathbb E_{U\sim\mu}
\Big[
\big(U^{(n)} \otimes U^{(n)}\big)
(S_z \otimes S_z)
\big(\overline U^{(n)} \otimes \overline U^{(n)} \big)
\Big].$
Because both $U^{\otimes n}$ and $S_z$ commute with all permutations of the qubits, $\Omega^{(2)}$ is invariant under the action of the symmetric group $S_n$. Consequently, $\mathcal M$ preserves the algebra of permutationally invariant operators.

The single-qubit Clifford group forms a unitary $2$-design. Therefore,
$\mathbb E_{U\sim\mu}
\big[
U Z U \otimes \overline U Z  \overline U 
\big]
= \sum_{\alpha=x,y,z}
\sigma_\alpha \otimes \sigma_\alpha/3.$ Using $S_z = \sum_iZ_i/2$ and linearity of the collective twirl, this identity extends to the collective setting, yielding 
$\Omega^{(2)}
= 
\sum_{\alpha=x,y,z}
S_\alpha \otimes S_\alpha/3,$
up to terms proportional to the identity which do not contribute to expectation-value reconstruction for traceless observables. It follows that  $\mathcal M$ is $\mathrm{SU}(2)$-covariant on the collective operator algebra
$\mathcal A_{\mathrm{coll}} := \mathrm{Alg}\{S_x,S_y,S_z\}.$ By $\mathrm{SU}(2)$ covariance and Schur--Weyl duality, $\mathcal M$ acts diagonally on irreducible tensor operators of fixed total spin. In particular, for any homogeneous polynomial $O_k \in \mathcal A_{\mathrm{coll}}$ of degree $k$, $\mathcal M(O_k) = \lambda_k\, O_k,$
where $\lambda_k>0$ depends only on $k$ and the system size $n$.
Since $\|S_\alpha\| = {\cal O}(n)$ and $\Omega^{(2)} = {\cal O}(n^2)$ in operator norm, the eigenvalues scale as
$\lambda_k = {\cal O}(n^k).$
Thus, $\mathcal M$ is invertible on $\mathcal A_{\mathrm{coll}}$.
We define the inverse shadow map on $\mathcal A_{\mathrm{coll}}$ by
 $\mathcal M^{-1}(O_k) := \lambda_k^{-1} O_k,$
extended linearly.
Given a single realization of the randomized measurement, the corresponding classical shadow is
$\hat\rho := \mathcal M^{-1}
\big(
\overline U^{(n)} S_z U^{(n)}
\big).$
For any observable $O \in \mathcal A_{\mathrm{coll}}$, we define the estimator
 $\hat O := \operatorname{Tr}(O\hat\rho).$ 
By construction,
$\mathbb E[\hat O] = \operatorname{Tr}(O\rho),$
so the estimator is unbiased.
The variance of the estimator is
$\mathrm{Var}(\hat O)
= \mathbb E[\hat O^2] - \operatorname{Tr}(O\rho)^2.$
As in the standard classical shadows analysis, the second moment is governed by a third-moment operator
\begin{equation}
\label{eq:Omega3}
\Omega^{(3)}
:= \mathbb E_{U\sim\mu}
\Big[
\big(U^{(n)} \otimes  U^{(n)} \otimes  U^{(n)}\big)
(S_z^{\otimes 3})
\big(\overline U^{(n)} \otimes \overline U^{(n)} \otimes \overline U^{(n)}\big)
\Big].
\end{equation}

Using permutation invariance and Clifford symmetry, $\Omega^{(3)}$ has operator norm $O(n^3)$. Together with
$\|\mathcal M^{-1}(O_k)\| = {\cal O}(n^k),$
this yields
$\mathrm{Var}(\hat O) = {\cal O}(n^{2k}) \text{(single shot)}.$
For $N$ independent repetitions,
$\mathrm{Var}(\hat O_N)
= {\cal O}\!\left({n^{2k}}/{N}\right).$
Collective Clifford classical shadows with collective $S_z$ readout define a measurement channel that is diagonal and invertible on the permutationally invariant operator algebra. Consequently, expectation values of collective observables that are polynomial in the collective spin operators can be estimated with variance scaling polynomially in system size.

\section{Theory of Experimental Shadow Metrology using NMR}
\label{App: Theory of experimental shadow metrology using NMR}



We begin with the single--qubit equatorial state
$\ket{\psi_\theta}
=\bigl(\ket{0}+e^{-i\theta}\ket{1}\bigr)/\sqrt{2},$
whose Bloch vector lies in the equatorial plane. Any single--qubit
Clifford unitary $U$ acts as a signed permutation of the Pauli axes and
therefore maps $\ket{\psi_\theta}$ into one of two canonical forms. If
the equatorial plane is preserved, the state remains equatorial, $U\ket{\psi_\theta}
=\bigl(\ket{0}+e^{-i(\pm\theta+k\pi/2)}\ket{1}\bigr)/{\sqrt{2}},
 k\in{0,1,2,3},$
corresponding to a discrete rotation and possible reflection of
$\theta$. If instead $U$ maps the equator to a meridian (as for Cliffords
containing a Hadamard), the amplitudes become unequal,
\begin{equation}
U\ket{\psi_\theta}
=e^{i\phi}
\Bigl(
\cos\frac{\pm\theta+k\pi/2}{2}\ket{0}
+i\sin\frac{\pm\theta+k\pi/2}{2}\ket{1}
\Bigr),
\end{equation}
with $\phi$ an irrelevant global phase. Thus Clifford operations never
introduce new functional dependence on $\theta$, but only discrete
shifts, reflections, or a halving of the phase.

This structure extends directly to multipartite GHZ--type states. For
three qubits, consider
$\ket{\Psi_\theta}
=\bigl(\ket{000}+e^{-i3\theta}\ket{111}\bigr)/{\sqrt{2}},$
and the collective action $U^{(3)}=U^{\otimes3}$. If $U$ preserves the
computational basis (i.e.\ $U=S^k$ up to phase), the output remains
GHZ--like,
$U^{(3)}\ket{\Psi_\theta}
=\bigl(\ket{000}
+e^{-i(3\theta-3k\pi/2)}\ket{111}\bigr)/{\sqrt{2}}.$
If $U$ maps $Z$ to $X$ or $Y$, writing
$U\ket{0}=(\ket{0}+e^{i\alpha}\ket{1})/\sqrt{2}$ and
$U\ket{1}=(\ket{0}-e^{i\alpha}\ket{1})/\sqrt{2}$, one finds
\begin{equation}
U^{(3)}\ket{\Psi_\theta}
=\frac{1}{2^{3/2}}
\sum_{x\in\{0,1\}^3}
\bigl[1+(-1)^{|x|}e^{-i3\theta}\bigr]
e^{i\alpha|x|}\ket{x},
\end{equation}
so that the phase $3\theta$ appears solely through parity--dependent
interference between even and odd excitation sectors.

The same reasoning applies to odd $n$--qubit GHZ states
$\ket{\Psi_\theta^{(n)}}
=
\bigl(\ket{0}^{\otimes n}+e^{-in\theta}\ket{1}^{\otimes n}\bigr)/{\sqrt{2}}, n\ \text{odd}.$
For $U^{(n)}=U^{\otimes n}$, Cliffords preserving $Z$ induce only a rigid
rotation of the collective phase,
$U^{(n)}\ket{\Psi_\theta^{(n)}}
=\bigl(\ket{0}^{\otimes n}
+e^{-i(n\theta-nk\pi/2)}\ket{1}^{\otimes n}\bigr)/{\sqrt{2}},$
while Cliffords mapping $Z$ to $X$ or $Y$ yield
\begin{equation}
U^{(n)}\ket{\Psi_\theta^{(n)}}
=\frac{1}{2^{n/2}}
\sum_{x\in{0,1}^n}
\bigl[1+(-1)^{|x|}e^{-in\theta}\bigr]
e^{i\alpha|x|}\ket{x}.
\end{equation}

In all cases, single--qubit Clifford operations preserve the functional
form of the embedded phase: $\theta$ enters only through the collective
factor $\exp({-in\theta})$, modified by discrete shifts and parity--dependent
signs. This constraint explains the robustness of phase information
under Clifford twirling and the persistence of sensitivity to $\theta$.

This structure is directly reflected in our NMR experiments, which were performed for phase–embedded states with
$n=1, 3, 5, 7, 9, 11, 13, 15$ qubits. In each case we prepare the GHZ–type state
$\ket{\Psi_\theta^{(n)}}
=\bigl(\ket{0}^{\otimes n}+e^{-in\theta}\ket{1}^{\otimes n}\bigr)/{\sqrt{2}},$ 
and apply collective single–qubit Clifford operations
$U^{(n)}=U^{\otimes n}$, implemented in NMR as global radio–frequency
pulses acting identically on all spins. As shown above, such collective
Cliffords cannot alter the functional dependence on $\theta$: they
only induce discrete phase shifts or convert the collective phase
$\exp({-in\theta})$ into parity–dependent interference between even– and
odd–excitation sectors.

We exploit this feature using Clifford lensing to coherently refocus the collective phase onto a single
spin. By applying three collective Clifford transformations, the parity–dependent interference terms
are converted into an effective phase kickback onto the first qubit
proportional to $n\theta$. The resulting transverse magnetization of
phase, while the Clifford twirl suppresses incoherent contributions.

The observed NMR signal across all odd system sizes $n=1$ to $15$
thus provides a direct experimental confirmation of the theoretical
result: collective Clifford dynamics preserve the encoded phase and,
through Clifford lensing, concentrate it into a low–weight single–qubit observable,
enabling robust and scalable extraction of $\theta$.
%


\end{document}